\pdfoutput=1

\documentclass[11pt]{article}
\usepackage{graphicx}
\usepackage{fullpage}
\usepackage{amsmath, amssymb, amsthm, bbm}
\usepackage{thmtools}
\usepackage{mathtools}
\usepackage[shortlabels]{enumitem}
\usepackage{hyperref}
\usepackage{stmaryrd}
\usepackage{tikz}
\usepackage{subfig}
\usepackage{microtype}
\hypersetup{breaklinks, urlcolor=blue, colorlinks, citecolor=green!50!black, linkcolor=blue!80!black}
\usepackage[capitalize,noabbrev,nameinlink]{cleveref}
\usepackage[ruled,vlined,linesnumbered]{algorithm2e}

\Crefname{algocf}{Algorithm}{Algorithms}
\Crefname{claim}{Claim}{Claims} 
\Crefname{problem}{Problem}{Problems}
\Crefname{fact}{Fact}{Facts}
\Crefname{observation}{Observation}{Observations}

\title{Efficient Matroid Intersection via a Batch-Update Auction Algorithm}
\author{
Joakim Blikstad\thanks{
        KTH Royal Institute of Technology \& Max Planck Institute for Informatics,
        \texttt{blikstad@kth.se}.
        Supported by the Swedish Research Council (Reg. No. 2019-05622) and the Google PhD Fellowship Program.}
\and
Ta-Wei Tu\thanks{
        Stanford University,
        \texttt{taweitu@stanford.edu}.
        Supported by a Stanford School of Engineering Fellowship, a Microsoft Research Faculty Fellowship, and NSF CAREER Award CCF-1844855.}
}
\date{}

\declaretheorem[numberwithin=section,refname={Theorem,Theorems},Refname={Theorem,Theorems}]{theorem}
\declaretheorem[numberlike=theorem]{lemma}

\declaretheorem[numberlike=theorem]{corollary}
\declaretheorem[numberlike=theorem]{claim}
\declaretheorem[numberlike=theorem]{fact}

\declaretheorem[numberlike=theorem,style=remark]{remark}

\newcommand{\I}{\mathcal{I}}
\newcommand{\M}{\mathcal{M}}
\newcommand{\R}{\mathbb{R}}
\newcommand{\eps}{\varepsilon}
\newcommand{\rank}{\mathsf{rank}}
\newcommand{\spn}{\mathsf{span}}

\newcommand{\poly}{\mathrm{poly}}
\newcommand{\defeq}{\stackrel{\mathrm{{\scriptscriptstyle def}}}{=}}

\begin{document}

\maketitle

\begin{abstract}
Given two matroids $\M_1$ and $\M_2$ over the same $n$-element ground set, the matroid intersection problem is to find a largest common independent set, whose size we denote by $r$.
We present a simple and generic auction algorithm that reduces $(1-\eps)$-approximate matroid intersection to roughly $1/\eps^2$ rounds of the easier problem of finding a maximum-weight basis of a single matroid.
Plugging in known primitives for this subproblem, we obtain both simpler and improved algorithms in two models of computation, including:
\begin{itemize}
\item
The first near-linear time/independence-query $(1-\eps)$-approximation algorithm for matroid intersection. Our randomized algorithm uses\footnote{Throughout the paper we use $\tilde{O}(\cdot)$ to hide $\mathrm{polylog}\,n$ factors.} $\tilde{O}(n/\eps + r/\eps^5)$ independence queries, improving upon the previous $\tilde{O}(n/\eps + r\sqrt{r}/{\eps^3})$ bound of Quanrud (2024).
\item
The first sublinear exact parallel algorithms for weighted matroid intersection, using
$O(n^{2/3})$ rounds of rank queries or $O(n^{5/6})$ rounds of independence queries. For the unweighted case, our results improve upon the previous $O(n^{3/4})$-round rank-query and $O(n^{7/8})$-round independence-query algorithms of Blikstad (2022).
\end{itemize}
\end{abstract}

\section{Introduction}

\paragraph{Matroid Intersection.}

A matroid is a fundamental combinatorial object capturing a notion
of \emph{independence} on the subsets of a so-called \emph{ground set} (see \cref{sec:prelim} for a formal definition).
For instance, the set of forests of a graph forms a matroid, where a subset of edges is called independent if and only if it has no cycles.
In this paper, we consider the well-studied optimization problem called \emph{matroid intersection}:
Given two matroids over the same ground set~$V$, we want to compute the largest set that is independent in both matroids.
The matroid intersection problem naturally generalizes a wide range of combinatorial optimization problems including bipartite matching, arborescences, colorful spanning trees, tree packing, etc.
Throughout this paper, we let $n$ denote the size of the ground set and $r$ the size of the answer (i.e., the largest common independent set).

\paragraph{Oracle Models.}
Since a matroid in general requires exponentially (in $n$) many bits to describe,
algorithms commonly assume oracle access to the matroids instead of an explicit representation.
The most standard query model is the \emph{independence} oracle that, when given a set $S \subseteq V$, returns whether $S$ is independent in the matroid or not.
Another well-studied oracle is the stronger \emph{rank} oracle that instead outputs the \emph{rank} of $S$, i.e., the size of the largest independent subset of $S$.
The efficiency of matroid intersection algorithms is measured by the number of queries (independence or rank). In most cases, our algorithms including, the running times are dominated by asking queries.

\paragraph{Parallel Algorithms.}
Apart from the standard sequential model of computation, we also consider parallel algorithms.
In this model, the algorithm is allowed to issue up to $\poly(n)$ queries to the oracle in parallel in each round.
These queries can only depend on the outcome of queries asked in previous rounds but not on each other.
The goal is to use as few number of rounds of queries---usually called the \emph{depth} or \emph{adaptivity} of the algorithm.

\paragraph{A New Auction Algorithm.}
Underlying our results is a simple and unified framework that reduces (additive) approximate matroid intersection to $O(1/\eps^2)$ rounds of the much easier problem of finding a maximum-weight basis of a \emph{single} matroid.
The reduction is inspired by an almost 30-year-old auction algorithm\footnote{Note that this is not for an auction problem involving matroids (e.g.,~\cite{BikhchandaniVSV11}), but rather an auction-style matroid intersection algorithm.} proposed in \cite{fujishige1995efficient,ShigenoI95}.
In the auction algorithm, elements in the matroid are assigned \emph{weights} that are adjusted iteratively throughout to facilitate finding a large common independent set.
After each adjustment, the algorithm needs to make queries to the oracle in order to maintain a maximum-weight basis for each matroid according to the new weights.
Our key observation is that, instead of adjusting the weight of one element at a time, we can batch and perform multiple (roughly $\eps n$ many) adjustments simultaneously. This means that we only need to recompute the two maximum weight bases for the two matroids a total of  $O(1/\eps^2)$ times.
This idea is reminiscent of a similar cross-paradigm auction algorithm for bipartite matching in the streaming and massively parallel computation settings by \cite{AssadiLT21}.

\paragraph{Our Results.}
With our auction algorithm we can plug in known primitives for finding a maximum-weight basis of a (single) matroid. This leads to improved and simple algorithms both in the sequential and parallel settings, both for rank and independence oracles.

Our first set of results are faster approximation algorithms for matroid intersection under the independence oracle model.
We show a simple algorithm that achieves an \emph{$\eps n$-additive} approximation in $O(n/\eps^2)$ independence queries.

\begin{restatable}{theorem}{Approximate}
  \label{thm:approximate}
  There is an $O(n/\eps^2)$-independence-query algorithm that computes an $S \in \I_1 \cap \I_2$ of size $|S| \geq r - \eps n$.
\end{restatable}

By combining \cref{thm:approximate} with the recent adaptive (and randomized) sparsification framework of \cite{Quanrud24} that reduces the ground set to be of size $\tilde{O}(r/\eps)$, we obtain
the \emph{first near-linear time/independence-query approximation scheme} for matroid intersection. This improves on previous state of the art which consists of query complexities $\tilde{O}(n\sqrt{r}/\eps)$ \cite{Blikstad21}, $\tilde{O}(n^2/(r\eps^2) + r^{1.5}/\eps^{4.5})$ \cite{ChakrabartyLS0W19}, and
$\tilde{O}(n/\eps + r^{1.5}/\eps^{3})$ \cite{Quanrud24}, the latter being near-linear only in the ``sparse'' regime where $r \le n^{2/3}$.

\begin{restatable}{theorem}{MultApproximate}
  \label{thm:mult-approximate}
  There is an randomized $O\left(\frac{n\log n}{\eps} + \frac{r\log^3 n}{\eps^5}\right)$-independence-query algorithm that computes an $S \in \I_1 \cap \I_2$ of size $|S| \geq (1-\eps)r$.
\end{restatable}

\begin{remark}
The state-of-the-art exact matroid intersection algorithm~\cite{BlikstadBMN21,Blikstad21} is obtained by first running an approximation algorithm to get a sufficiently large solution and then falling back to finding augmenting paths one by one.
There are currently two tight bottlenecks in improving the $\tilde{O}(n^{7/4})$ bound of \cite{Blikstad21}, the first one being the approximate algorithm and the second being the augmentation algorithm.
Our improved approximation algorithm resolves the first bottleneck, leaving only the second.\footnote{Here we consider mostly the regime where $r \approx n$. When $r \ll n$, the $\tilde{O}(n/\eps + r^{1.5}/\eps^3)$ time algorithm of \cite{Quanrud24} could similarly lead to a faster exact algorithm if there was a faster augmentation algorithm.}
This means that, if one can find a single augmenting path faster than the  $\tilde{O}(n\sqrt{r})$-independence-query augmentation algorithm of \cite{BlikstadBMN21}, it would now, when combined with our \cref{thm:approximate}, directly imply a faster exact matroid intersection algorithm.
\end{remark}

Our auction framework also works well in the parallel setting. When combined with standard parallel primitives, we obtain faster parallel approximation algorithms for matroid intersection, which also implies faster parallel \emph{exact} algorithms.
When $r = \Theta(n)$, our algorithms take $O(n^{2/3})$ rounds of rank queries or $O(n^{5/6})$ rounds of independence queries, improving on the $O(n^{3/4})$ rank-query and $O(n^{7/8})$ independence-query bounds of \cite{Blikstad22}.
This brings us slightly closer towards the information-theoretic lower bound of $\tilde{\Omega}(n^{1/3})$ rounds of independence queries \cite{KarpUW85} or rank queries \cite{ChakrabartyCK21}.

\begin{restatable}{theorem}{Parallel}
  \label{thm:parallel}
  There is an $O(n^{1/3}r^{1/3})$-round parallel rank-query algorithm and an $O(n^{1/3}r^{1/2}\log(n/r))$-round parallel independence-query algorithm for matroid intersection.
\end{restatable}

In the parallel setting, the auction algorithm also fits well into the \emph{weight-splitting} approach of \cite{fujishige1995efficient,ShigenoI95} for solving \emph{weighted} matroid intersection. This allows us to obtain the first sublinear round parallel algorithms for the problem, matching our unweighted results up to a $\log (Wr)$-overhead (where $W$ is the largest weight). Previously, no non-trivial bound was known for the weighted case.

\begin{restatable}{theorem}{ParallelWeighted}
  \label{thm:parallel-weighted}
  There is an $O(n^{1/3}r^{1/3}\log (Wr))$-round parallel rank-query algorithm and an $O(n^{1/3}r^{1/2}\log (Wr)\log(n/r))$-round parallel independence-query algorithm for weighted matroid intersection.
\end{restatable}

While the efficiencies of the above theorems are stated in terms of the number of independence queries or the query depth, it will be clear from the description of the algorithms that the (sequential) running times or the (parallel) work depths of them are nearly\footnote{As we will see later in \cref{sec:parallel}, one component of our parallel algorithms involves building the exchange graph and then finding an augmenting path in it. While the graph can be constructed in a single round of queries, solving the reachability problem in a directed graph requires $O(\log n)$ rounds of adaptivity. This incurs an $O(\log n)$ parallel depth overhead to the query depth.} dominated by the former quantities.
Our algorithms also trivially work in the recently introduced \emph{dynamic} oracle model~\cite{BlikstadMNT23}.

\paragraph{Related Work.}
In the 1960s, Edmonds showed the first polynomial query algorithm to solve matroid intersection. Since then, there has been a long line of research
(e.g., \cite{edmonds1968matroid, edmonds1970submodular,aignerD,Lawler75,cunningham1986improved,fujishige1995efficient,ShigenoI95,LeeSW15,nguyen2019note,BlikstadBMN21,Quanrud24})
 culminating in a $\tilde{O}(n\sqrt{r})$-rank query algorithm \cite{ChakrabartyLS0W19} and a $\tilde{O}(nr^{3/4})$-independence query algorithm \cite{Blikstad21}. A crucial ingredient in these exact algorithms is the use of efficient $(1-\eps)$-approximation algorithms: The best-known bounds are $\tilde{O}(n/\eps)$ rank queries \cite{ChakrabartyLS0W19}, $\tilde{O}(n\sqrt{r}/\eps)$ independence queries \cite{Blikstad21}, or $\tilde{O}(n/\eps + r\sqrt{r}/\eps^{3})$ independence 
 queries \cite{Quanrud24}. Virtually all of these efficient algorithms work by identifying ``blocking flows'' in the so-called ``exchange graphs'' (whose structure turns out to be more intricate than its counterparts in, e.g., bipartite matching). 
 We take a completely different ``auctioning'' approach similar to \cite{fujishige1995efficient,ShigenoI95}, which makes our algorithms simpler in comparison (see \cref{alg:auction}).

Parallel matroid intersection has also been studied, see e.g., \cite{ChakrabartyCK21,GhoshGR22,Blikstad22}. There are $\Omega(n^{1/3})$-round lower bounds proven in \cite{KarpUW85} for independence-query algorithms and \cite{ChakrabartyCK21} for rank queries.
Implicit in Edmonds' early work is a trivial $O(n)$-round algorithm. The only previous sublinear-round algorithms are due to \cite{Blikstad22} who gave an $O(n^{3/4})$-round rank-query algorithm and an $O(n^{7/8})$ independence-query algorithm for the unweighted case.
It is worth mentioning that \cite[Section~7]{ChakrabartyLS0W19} implicitly obtained a similar result as ours that reduces (additive) approximate matroid intersection to $O(1/\eps^2)$ rounds of maximum-weight basis of a single matroid via the Frank-Wolfe convex optimization algorithm.
However, their reduction only gives a \emph{fractional} approximate solution.
Since there is no known efficient rounding algorithm for matroid intersection (sequential nor parallel), their result does not imply improved parallel exact algorithms like ours.

\section{Preliminaries} \label{sec:prelim}

As is common in matroid literature, for a set $A$ and elements $a\in A, b\notin A$ we will use $A-a$ and $A+b$ to denote $A\setminus \{a\}$ and $A \cup \{b\}$ respectively.

\paragraph{Matroid.} A \emph{matroid} $\M = (V,\I)$ consists of a ground set $V$ with a notion of independence $\I \subseteq 2^{V}$.
The notion of independence $\I$ must satisfy three properties: (i) $\emptyset \in \I$; (ii) \emph{downwards-closure} meaning that if $S\in \I$ and $S'\subseteq S$, then $S'\in \I$; and (iii) \emph{the exchange property} meaning that if $S, S'\in \I$ and $|S| < |S'|$, then there exists some $x\in S'\setminus S$ so that $S+x\in \I$.

A subset $S\subseteq V$  is called \emph{independent} if $S\in \I$; it is called a \emph{basis} if it is a maximally independent set, that is if $S+x\not\in \I$ for all $x\in V\setminus S$.
The rank of a subset $S\subseteq V$ is the size of the largest independent set inside $S$, denoted by $\rank(S)$.
As a consequence of the exchange property, all bases of a matroid have the same size---$\rank(V)$, also written $\rank(\M)$---which we also call the rank of the matroid $\M$.
When clear from context, we let $n$ denote the size of the ground set $V$.

\paragraph{Maximum Weight Basis.}
Given a weight function $w: V\to \R$ on the 
ground set, the \emph{maximum weight} basis of a matroid $\M$ is the basis $B$ which maximizes $\sum_{e\in B} w(e)$. To find a maximum weight basis of a matroid in the sequential setting, the simple greedy algorithm that scans elements in descending order of weights suffices.
\begin{fact}
  \label{fact:sequential-max-basis}
  Given weights $w(e)$ for each $e \in V$, there is an $O(n)$-independence-query algorithm that finds a $w$-maximum basis in a matroid $\M$.
\end{fact}

Finding a basis of a matroid in the \emph{parallel} setting has been studied before in \cite{KarpUW85},
who showed a $1$-round rank-query algorithm and an $\tilde{O}(\sqrt{\rank(\M)})$-round independence-query algorithm, together with an $\Omega(\rank(\M)^{1/3})$-round independence-query lower bound.
It is easy to convert these algorithms to find a \emph{maximum-weight} basis instead, in the same number of rounds.
For completeness we defer a proof sketch to \cref{appendix:basis}.

\begin{restatable}{lemma}{ParallelMaxBasis}
  \label{lemma:parallel-max-basis}
  Given a weight $w(e)$ for each $e \in V$, there is a $1$-round rank-query algorithm and an $O(\sqrt{\rank(\M)}\log(n/\rank(\M)))$-round independence-query algorithm that finds a $w$-maximum basis in a matroid $\M$.
\end{restatable}

\paragraph{Matroid Intersection.} In the matroid intersection problem one is given two matroids $\M_1 = (V,\I_1)$ and $\M_2 = (V,\I_2)$ on the same ground sets but with different notions of independence. The goal is to compute the largest \emph{common independent set} $S\in \I_1\cap \I_2$. In the \emph{weighted} matroid intersection problem, one is additionally given a weight function $w:V\to \mathbb{Z}$, and instead wants to find the common independent set $S$ of maximum weight $\sum_{e\in S}w(e)$.
When discussing matroid intersection, we let $n$ denote $|V|$ and $r$ the size of the answer;\footnote{Sometimes in literature $r$ denotes the rank of the input matroids instead; this is essentially the same as one can truncate the input matroids to have rank roughly the size of answer.}
and for the weighted version $W\defeq \max_{e\in V} |w(e)|$.
We make use of the following standard primal-dual formulation of matroid intersection.

\begin{theorem}[Matroid Intersection Theorem~\cite{edmonds1979matroid}]
\label{thm:matroid-intersection-dual}
For matroids $\M_1=(V,\I_1)$ and $\M_2=(V,\I_2)$ it holds that
$$\max_{S\in \I_1\cap \I_2}|S| = \min_{A,B\subseteq V, A\cup B = V} \rank_1(A) + \rank_2(B).$$
\end{theorem}

\paragraph{Augmentation.}
In the parallel setting, the well-known \emph{exchange graph} can be constructed in a single round of $O(nr)$ many independence (or rank) queries.
Such a graph allows us to augment a common independent set, increasing its size by one.

\begin{fact}
  \label{fact:augment}
  Given an $S \in \I_1 \cap \I_2$, in a single round of independence or rank queries, we can compute an $S^\prime \in \I_1 \cap \I_2$ of size $|S^\prime| = |S| + 1$ or decide that $S$ is of maximum possible size.
\end{fact}

\section{The Auction Algorithm}

Both our parallel and sequential algorithms are obtained by a more efficient implementation of the auction algorithm originally proposed by \cite{fujishige1995efficient,ShigenoI95} for weighted matroid intersection algorithms.
The auction algorithm assigns two weights $w_1(e)$ and $w_2(e)$ to each element $e\in V$, arbitrarily initialized to $0$.
The algorithm maintains a $w_1$-maximum basis $S_1$ in the first matroid $\M_1$ and a $w_2$-maximum basis $S_2$ in the second matroid $\M_2$. Note that $S \defeq S_1\cap S_2$ is a common independent set, so if $S$ is large, we are done. Otherwise, there must be many ``problematic'' elements in $S_1\setminus S_2$.
For each such problematic element $e\in S_1\setminus S_2$, we do one of two things: (i) discourage $e$ from being in $S_1$ by decreasing $w_1(e)$ by $1$, or (ii) encourage $e$ to be in $S_2$ by increasing $w_2(e)$ by $1$. The hope is that $e$ ceases to be problematic and that $S_1\setminus S_2$ eventually becomes small.
We alternate between adjusting $w_1$ and $w_2$,  which guarantees that $w_1(e) + w_2(e)$ is always either $0$ or $1$.
Each time one of $e$'s weights is adjusted we increase its ``price'' $p(e)$, and later argue that once the price exceeds $2/\eps$, we may simply ignore $e$.\footnote{One may also interpret this as a push-relabel-style algorithm, where the weight of an element corresponds to its label.} 

Our key modification to the auction algorithm is that we process all the weight adjustments of the currently problematic elements as a batch, instead of handling them one at a time. As long as there are at least $\Delta$ (a parameter to the algorithm; think of it as being roughly $\eps n$) remaining problematic elements, we increase the price of at least $\Delta$ elements. This cannot happen too many times, and we can afford to completely recompute $S_1$ and $S_2$ from scratch in each round. On the other hand, if there are fewer than $\Delta$ problematic elements we stop early and prove that this only incurs an additional additive $\Delta$ approximation error. The pseudocode is in \cref{alg:auction}.

\begin{algorithm}[!htb]
  \caption{Auction algorithm for matroid intersection} \label{alg:auction}
  
  \SetEndCharOfAlgoLine{}

  \SetKwInput{KwData}{Input}
  \SetKwInput{KwResult}{Output}
  \SetKwInOut{State}{global}
  \SetKwProg{KwProc}{function}{}{}

  \KwData{Two matroids $\M_1$ and $\M_2$; parameters $\eps\in (0,1)$ and $\Delta > 0$.}
  \KwResult{A common independent set $S$ of size $|S| \ge r - (\eps r + \Delta)$.}

  $p(e), w_1(e), w_2(e) \gets 0$ for all $e \in V$.\; \label{line:init}
  $S_1 \gets$ a basis in $\M_1$ and $S_2 \gets$ a basis in $\M_2$.\;
  \While{true} { \label{line:iter}
    $X \gets \{e \in S_1 \setminus S_2: p(e) < 2/\eps\}$. \label{line:X}\;
    \lIf{$|X| < \Delta$} {
      \textbf{return} $S \defeq S_1 \cap S_2$.
      \label{line:return}
    }
    \For{$e \in X$} {
      $p(e) \gets p(e) + 1$.\;
      \lIf{$w_1(e) + w_2(e) = 0$} {
          $w_2(e) \gets w_2(e) + 1$.\label{line:adjust-1}
      }
      \lElse{
        $w_1(e) \gets w_1(e) - 1$. \label{line:adjust-2}
      }
    }
    \For{$i \in \{1, 2\}$} {
      $S_i \gets$ a $w_i$-maximum basis in $\M_i$\; \label{line:find-basis}
      (breaking ties by preferring elements that were in the old $S_i$).\;
    }
  }
\end{algorithm}

  The following claim is useful for establishing the correctness of our algorithm.
  Similar claims appeared also in \cite{fujishige1995efficient,ShigenoI95}.
  \begin{claim}
    \label{claim:s2-minus-s1}
    We have $p(S_2 \setminus S_1) = 0$ after each iteration of the algorithm.
    In particular, for all $e\in S_2\setminus S_1$, we have $w_2(e) = 0$.
  \end{claim}
  \begin{proof}
    We are going to show that $S_2 \setminus S_1$ is a decremental set, i.e., we only move elements out of $S_2 \setminus S_1$ but never into.
    Consider a single phase of adjustments.
    Let $X_2 \subseteq X$ be elements with $w_1(e) + w_2(e) = 0$ (for which we will increase the $w_2$-weights), and let $X_1 \defeq X \setminus X_2$ be those we will decrease their $w_1$-weights.
    Let us for clarity denote by $S_i^{\mathrm{new}}$ the basis $S_i$ after Line~\ref{line:find-basis} and by $S_i^{\mathrm{old}}$ the one before it.
    By the tie-breaking, $S_i^{\mathrm{new}}$ prefers elements that were in $S_i^{\mathrm{old}}$, so it follows from standard exchange properties of matroids that $S_1^{\mathrm{new}} = (S_1^{\mathrm{old}} \setminus X_1^\prime) \cup Y_1$ and $S_2^{\mathrm{new}} = (S_2^{\mathrm{old}} \setminus Y_2) \cup X_2^\prime$ for some $X_1^\prime \subseteq X_1$, $Y_1 \subseteq V \setminus S_1^{\mathrm{old}}$, $X_2^\prime \subseteq X_2$, and $Y_2 \subseteq S_2^{\mathrm{old}}$.
    As alluded to in the beginning, we will show that no elements ``move'' to $S_2 \setminus S_1$ when we go from $(S_1^{\mathrm{old}}, S_2^{\mathrm{old}})$ to $(S_1^{\mathrm{new}}, S_2^{\mathrm{new}})$.
    Note that as $X_2^\prime$ and $X_1^\prime$ are disjoint, all elements in $X_2^\prime$ have their $w_1$-weight intact and thus stay in $S_1$.
    Similarly, all elements in $X_1^\prime$ have their $w_2$-weight intact and thus stay outside $S_2$.
    We can now do a case analysis:
    (1) $X_1^\prime$ moves to $V \setminus (S_1^{\mathrm{new}} \cup S_2^{\mathrm{new}})$; 
    (2) $Y_1$ moves to $S_1^{\mathrm{new}}$;
    (3) $X_2^\prime$ moves to $S_2^{\mathrm{new}} \cap S_1^{\mathrm{new}}$;
    (4) $Y_2$ moves to $V \setminus S_2^{\mathrm{new}}$.
    This proves the claim.
  \end{proof}

We can now lower bound the size of our output $S$.
\begin{lemma} \label{lemma:meta}
  Given $\eps$ and $\Delta$, \cref{alg:auction} outputs an $S \in 
\I_1\cap \I_2$ of size $|S| \geq r - (\eps r + \Delta)$.
\end{lemma}
There are many ways to prove \cref{lemma:meta}; we will do so 
by showing the following stronger \cref{lemma:dual} which explicitly extracts a dual solution from the algorithm (we will also need this dual solution later in \cref{sec:mult-approx}). \cref{lemma:meta} then follows directly from \cref{lemma:dual} together with
the matroid intersection theorem (\cref{thm:matroid-intersection-dual}).
  
\begin{lemma}
\label{lemma:dual}
  After running \cref{alg:auction}, we can compute in $O(n)$ time (without additional queries) two sets $A, B \subseteq V$ such that (i) $A \cup B = V$ and (ii) $\rank_1(A) + \rank_2(B) \leq |S| + (\eps r + \Delta)$.
\end{lemma}

\begin{proof}
  Let $X$ be the set defined on Line~\ref{line:X} in \cref{alg:auction} when the algorithm returns a common independent set $S \defeq S_1 \cap S_2$ on Line~\ref{line:return}.
  We choose $t \in \{0, 1, \ldots, 1/\eps - 1\}$ minimizing the size of the set $S^{(t)} \defeq \{e \in S: w_1(e) = -t\}$.
  Since there are $1/\eps$ choices of $t$ and $|S|\le r$, we have $|S^{(t)}| \leq \eps r$.
  We then let $A \defeq \{v \in V: w_1(v) \geq -t\}$ and $B \defeq \{v \in V: w_2(v) \geq 1 + t\}$ be the output of this lemma.
  
  We first show that $A \cup B = V$.
  Indeed, $w_1(e) + w_2(e)\in\{0, 1\}$ for all $e\in V$, so either $w_1(e)\ge -t$ (and $e\in A$), or $w_1(e) \le -t-1$ in which case $w_2(e)\ge -w_1(e) \ge 1 + t$ (and $e\in B$).
  
  We now show that $\rank_1(A) + \rank_2(B) \leq |S| + \eps r + \Delta$.
  Since $S_1$ is a $w_1$-maximum basis in $\M_1$ and $A$ consists of all elements with $w_1$-weight at least some threshold, we know that $S_1\cap A$ is a largest independent set (of $\M_1$) inside $A$. Hence $\rank_1(A) = |S_1\cap A|$, and by a symmetric argument, $\rank_2(B) = |S_2 \cap B|$.
  This means that $\rank_1(A) + \rank_2(B) = |S_1\cap A| + |S_2\cap B| = |S\cap A| + |(S_1\setminus S_2)\cap A| + |S\cap B| + |(S_2\setminus S_1)\cap B|$.
  In fact, we have $(S_2\setminus S_1)\cap B = \emptyset$, since
  $w_2(e)>0$ for each element $e\in B$, and $w_2(e') = 0$ for each element $e'\in S_2\setminus S_1$ by \cref{claim:s2-minus-s1}. Moreover, $(S_1\setminus S_2)\cap A\subseteq X$, 
  since $w_1(e) \geq -(1/\eps) + 1$ for each element $e\in A$, and, by construction, $X$ contains exactly the elements of $S_1\setminus S_2$ with $w_1$-weight larger than $-(1/\eps-1)$.

  Hence,
  $\rank_1(A) + \rank_2(B) \le |S\cap A| + |S\cap B| + |X| = |S| + |S\cap A\cap B| + |X|$. We know that $|X|<\Delta$ at the end of \cref{alg:auction}. We also know that $e \in S\cap A\cap B$ implies that $w_1(e) \ge -t$ and $w_2(e)\ge 1+t$; at the same time $w_1(e)+w_2(e)\le 1$ so it must be the case that all these three inequalities are tight. In particular $w_1(e) = -t$ and $e\in S$, so we can say $S\cap A\cap B \subseteq S^{(t)}$, where $|S^{(t)}|\le \eps r$.
  Concluding, we have $\rank_1(A)+\rank_2(B)\le |S| + \eps r + \Delta$.
\end{proof}

\paragraph{Running Time.} We now argue that the loop in \cref{alg:auction} does not run so many times. For example, if $\Delta = \eps n$, the following lemma shows it only runs $O(1/\eps^{2})$ times.

\begin{lemma} \label{lemma:iter}
  Given $\eps$ and $\Delta$, the while loop on Line~\ref{line:iter} in \cref{alg:auction} has at most $O\left(\frac{n}{\eps \Delta}\right)$ iterations.
\end{lemma}

\begin{proof}
  Let $\Phi \defeq \sum_{v \in V}p(v)$  be a potential which is initially $0$ and at most $2n/\eps$.
  Note that we increase $\Phi$ by exactly $|X|$ in each iteration, and since the algorithm terminates whenever $|X| < \Delta$, the potential increase per iteration is at least $\Delta$.
  This shows that there are at most $O\left(\frac{n}{\eps \Delta}\right)$ iterations.
\end{proof}

\subsection{A Simple Independence-Query Approximation Algorithm} \label{sec:additive}

We begin with a sequential implementation of our algorithm in the independence query model. By \cref{fact:sequential-max-basis}, a maximum weight basis can be found using a simple greedy algorithm with $O(n)$ independence queries.\footnote{Note that \cref{alg:auction} needs tie-breaking by preferring elements that were in the old basis. This can be done by slightly increasing the weight (say by $\eps/2$) of the preferred elements.} If we plug this into \cref{alg:auction}, we get the following running time.

\begin{corollary}
  One can implement \cref{alg:auction} using $O\left(\frac{n^2}{\eps\Delta}\right)$ independence queries.
\end{corollary}

To get $\eps n$ additive error, we may simply set $\Delta = \eps n$. This immediately proves \cref{thm:approximate}.

\Approximate*

On the other hand, if we want a $(1-\eps)$-multiplicative error, we may set $\Delta = \eps r$. If we do not know~$r$, one can first find a $2$-approximation of $r$ by a simple greedy $O(n)$-query algorithm %
and use this to set the parameter $\Delta$. We thus get the following $(1-\eps)$-approximation algorithm.

\begin{restatable}{theorem}{MultApproximateSimple}
  \label{thm:mult-approximate-simple}
  There is an $O\left(\frac{n^2}{r\eps^2}\right)$-independence-query algorithm that computes a $S \in \I_1 \cap \I_2$ of size $|S| \geq (1-\eps)r$.
\end{restatable}

\begin{remark}
  We note that when $r = \Theta(n)$ (and $\eps$ constant), the above algorithm runs in linear time. In fact, for this regime of $r$, it is the first (near-)linear time $(1-\eps)$-approximation algorithm for matroid intersection. Later in \cref{sec:mult-approx} we combine our algorithm with a recent specification approach \cite{Quanrud24} that essentially lets us assume $n = \tilde{O}(r/\eps)$, and hence we obtain a near-linear time algorithm for the full range of $r$.
\end{remark}

\subsection{Parallel Exact Algorithms} \label{sec:parallel}

We now look at parallel implementations of the auction algorithm.
By trying each value of $\tilde{r} \in \{0, \ldots, n\}$ in parallel and running the algorithm on the two input matroids truncated to rank $\tilde{r}$, we may assume we are in the case where both $\M_1$ and $\M_2$ have rank $r$ and share a common basis (this happens when $\tilde{r} = r$).
In particular, \cref{lemma:parallel-max-basis} tells us that one can find a maximum weight basis in either 1 round of rank queries or $O(\sqrt{r}\log(n/r))$ rounds\footnote{Technically, for executions with $\tilde{r} > r$, the matroids may have rank larger than $r$ and thus \cref{lemma:parallel-max-basis} takes $O(\sqrt{\rank(\M)}) > O(\sqrt{r})$ rounds of independence queries. Nevertheless, we can terminate all executions with $\tilde{r} > 2r$ when that of $\tilde{r} = 2r$ terminates and fails to produce a common independent set of size $2r$. This ensures that the parallel running time and query complexity of our algorithms match the execution with $\tilde{r} \leq 2r$. Thus, for simplicity of presentation we assume \cref{lemma:parallel-max-basis} takes $O(\sqrt{r})$ rounds of indepedence queries.} of independence queries.
Similar to \cref{sec:additive}, we can break ties by increasing the weights of elements in the old $S_i$ by a small amount.
Thus, together with \cref{alg:auction} and \cref{lemma:iter} we get the following parallel approximation algorithm.

\begin{corollary}
  \label{cor:parallel-impl}
  There is a parallel implementation of \cref{alg:auction} that uses $O\left(\frac{n}{\eps \Delta}\right)$ rounds of rank queries or $O\left(\frac{n\sqrt{r}\log(n/r)}{\eps \Delta}\right)$ rounds of independence queries.
\end{corollary}

\paragraph{From Approximate to Exact.}
We can augment a common independent set $S$ of size $|S| = r - s$ for $s$ times until it becomes an optimal solution of size $r$. Each augmentation takes only a single round in the parallel model (by \cref{fact:augment}), which gives us the following exact parallel algorithms.

\Parallel*

\begin{proof}
  Recall that we may assume we know the correct value of $r$ by trying all $\tilde{r} \in \{0, \ldots, n\}$ in parallel.
  For the rank-query algorithm, we run \cref{cor:parallel-impl} with $\eps = n^{1/3}r^{-2/3}$ and $\Delta = \eps r = n^{1/3}r^{1/3}$.
  This takes $O\left(\frac{n}{\eps\Delta}\right) = O(n^{1/3}r^{1/3})$ rounds of rank queries.
  By \cref{lemma:meta}, we get an $S \in \I_1 \cap \I_2$ of size $|S| \geq r - (\eps r + \Delta) \geq r - O(n^{1/3}r^{1/3})$.
  We then repeatedly augment $S$ using \cref{fact:augment} until it has size $r$.
  This takes $O(n^{1/3}r^{1/3})$ rounds of rank queries, and the overall algorithm runs in $O(n^{1/3}r^{1/3})$ rounds as well.

  For the independence-query algorithm, we use a different set of parameters $\eps = n^{1/3}r^{-1/2}$ and $\Delta = \eps r = n^{1/3}r^{1/2}$.
  Using the same argument as above, we can get an $S \in \I_1 \cap \I_2$ of size $|S| \geq r - O(n^{1/3}r^{1/2})$ in $O(n^{1/3}r^{1/2})$ rounds of independence queries, from which we can augment to an optimal solution also in $O(n^{1/3}r^{1/2})$ using \cref{fact:augment}.
  This proves the theorem.
\end{proof}

\subsection{Near-Linear Time Approximation Algorithm}
\label{sec:mult-approx}

Our result in \cref{sec:additive} only gives an $\eps n$-additive approximation or a linear time algorithm when $r = \Theta(n)$.
To further obtain a $(1-\eps)$-multiplicative (or, equivalently, $\eps r$-additive) approximation and prove \cref{thm:mult-approximate}, we apply a recent sparsification framework of \cite{Quanrud24} which we state below.

\begin{restatable}[modified from {\cite{Quanrud24}}; sketched in \cref{appendix:sparsification}]{lemma}{Sparsification}
  Suppose there is an oracle $\mathcal{A}$ that, when given as input a subset $U \subseteq V$, outputs using $T_{\mathcal{A}}(|U|)$ independence queries an $S_U \in \I_1 \cap \I_2$ and two sets $(A_U, B_U)$ such that
  \begin{itemize}
    \item $S_U, A_U, B_U \subseteq U$,
    \item $A_U \cup B_U = U$, and
    \item $\rank_1(A_U) + \rank_2(B_U) \leq |S_U| + \Delta_{\mathcal{A}}$.
  \end{itemize}
  Then, there is an randomized $O\left(\frac{n \log n}{\eps} + \frac{T_{\mathcal{A}}(\Theta(r\log n/\eps))\log n}{\eps}\right)$-independence-query algorithm that computes an $S \in \I_1 \cap \I_2$ of size $|S| \geq r(1-O(\eps)) - \Delta_{\mathcal{A}}$ with high probability.
  \label{lemma:sparsification}
\end{restatable}

We are now ready to prove the near-linear time/independence-query $(1-\eps)$-approximation algorithm in \cref{thm:mult-approximate}.
In \cref{lemma:dual} we already showed how to construct the dual certificate $(A_U, B_U)$ needed by \cref{lemma:sparsification}, so what remains is to plug in our \cref{alg:auction} into the the sparsification framework. 

\MultApproximate*

\begin{proof}
We need to show how to use \cref{alg:auction} to implement the oracle $\mathcal{A}$ of \cref{lemma:sparsification}.
For the given $\eps > 0$, we will set $\Delta = O(\eps r)$ when running \cref{alg:auction}.
  Again, without knowing $r$, this can be achieved using the $2$-approximation of it obtainable in $O(n)$ independence queries. %
  Note that this is the $r$ of the original problem and \textbf{not} the size $r_U$ of the largest common independent set in the subproblem restricted to some set $U$. Using $r_U$ instead would not work, since it can be much smaller than $|U| \approx r\log n/\eps$ and that would make the running time and query complexity too high.

By the choice of $\Delta$, $\mathcal{A}$ will use $T_{\mathcal{A}}(|U|) = O(\frac{|U|^2}{\eps \Delta}) = O(\frac{|U|^2}{\eps^2 r})$ independence queries by \cref{lemma:iter} and \cref{fact:sequential-max-basis} (as in the proofs of \cref{thm:approximate,thm:mult-approximate-simple}). Moreover, when run on the set $U$, $\mathcal{A}$ will return a common independent set $S_U\subseteq U$ and the dual certificate $(A_U,B_U)$ (with $A_U\cup B_U = U$) so that $\rank_1(A_U) + \rank_2(B_U) \le |S_U| + \Delta_{\mathcal{A}}$, where $\Delta_{\mathcal{A}} =  \eps r_U + \Delta = O(\eps r)$ is the additive error by \cref{lemma:dual}.

  Plugging $\mathcal{A}$ into the adaptive sparsification framework of \cref{lemma:sparsification}, we get an approximation algorithm that produces a common independent set $S$ of size $|S|\ge r - O(\eps r)$, using only $O\left(\frac{n \log n}{\eps} + \frac{(r\log n/\eps)^2\log n}{\eps^3 r}\right) = O\left(\frac{n\log n}{\eps} + \frac{r\log^3 n}{\eps^5}\right)$ independence queries. If we want $|S| \ge (1-\eps)r$, we can simply decrease $\eps$ by an appropriate constant which does not change the asymptotic running time and query complexity.
\end{proof}

\section*{Acknowledgements}
We thank Danupon Nanongkai and Aaron Sidford
for insightful discussions and their valuable comments throughout the development of this work.

\bibliography{reference}

\appendix

\section{Proof of \texorpdfstring{\cref{lemma:parallel-max-basis}}{Lemma 2.2}} \label{appendix:basis}

\begin{lemma}[\cite{KarpUW85,Blikstad22}]
  There is a $1$-round rank-query algorithm and an
  \[ O\left(\sqrt{\rank(\M)}\log(n/\rank(\M)\right) \]
  round independence-query algorithm that finds a basis in a matroid $\M$.
  \label{lemma:parallel-basis}
\end{lemma}
\begin{proof}[Proof Sketch.]
  Let $r \defeq \rank(\M)$.
  For completeness we sketch the algorithms here, also since the previous papers focused on when $r = \Theta(n)$ and thus only gave a $O(\sqrt{n})$-round bound for the independence query algorithm. See \cite{KarpUW85} and \cite[Section~3.1]{Blikstad22} for more details.
  
  Let $e_1, \ldots, e_n$ be the elements and let $V_i = \{e_1,e_2,\ldots, e_i\}$ be the prefix sets.
  The rank query algorithm can simply query all prefix sets in a single round in parallel.
  It is easy to show that $S = \{e_i : \rank(V_i) > \rank(V_{i-1})\}$ is a max-weight basis.

  The independence-query algorithm will keep track of an initially empty independent set $S$.
  The goal in each round is to either increase the size of $S$ by roughly $\sqrt{r}$, or remove roughly $n/\sqrt{r}$ many elements without changing the span;
  first of which can happen at most $\sqrt{r}$ times and the second at most $\sqrt{r}$ times until $n$ has decreased by a constant factor.
  In each round, split the elements of $V\setminus S$ into buckets $F_1, F_2, \ldots F_{n/\sqrt{r}}$, each of size roughly $\sqrt{r}$.
  Define $F_{i,k}$ to be the first $k$ elements of $F_i$.
  The algorithm queries $S+F_{i,k}$ for all $i$ and $k$.
  In case some set $S + F_{i}$ was independent, set $S\gets S+F_{i}$ (and the size of $S$ increased by $\sqrt{\rank(\M)}$).
  Otherwise, for each $i$ there must have been a $k$ so that $S+F_{i,k-1}$ is independent but $S+F_{i,k}$ is not.
  Then it is safe to remove (from $V$) the $k$-th element of $F_{i}$, since it is spanned by some elements which remain.
  This can be done for all $i$, after which we can remove $n/\sqrt{r}$ elements from the ground set.

  Since after $O(\sqrt{r})$ rounds the value of $n$ decreases by a constant factor, after $O(\sqrt{r}\log(n/r))$ rounds it becomes $O(r)$, at which point we can finish the algorithm by running the algorithm of \cite{KarpUW85} in $O(\sqrt{r})$ rounds.
\end{proof}

\ParallelMaxBasis*

\begin{proof}
  Let $e_1, \ldots, e_n$ be the elements sorted by non-increasing weights ($w(e_1) \ge w(e_2) \ge \ldots \ge w(e_n)$), and let $V_i = \{e_1,e_2,\ldots, e_i\}$ be the prefix sets.
  In parallel, for each $i$, find the rank of the sets $V_i$.
  For the independence-query algorithm, this can be done using $O(n)$ parallel instances of the $O(\sqrt{\rank(\M)}\log(n/\rank(\M)))$-round algorithm 
  of \cref{lemma:parallel-basis}. Let $S = \{e_i : \rank(V_i) > \rank(V_{i-1})\}$, and note that $S$ is a maximum-weight basis.
\end{proof}

\section{An Additive-Approximate Version of \texorpdfstring{\cite{Quanrud24}}{[Qua24]}} \label{appendix:sparsification}

In this section we sketch how \cref{lemma:sparsification} can be proven.
We use the exact same algorithm as \cite{Quanrud24} except that we use an additive approximate subroutine as given to \cref{lemma:sparsification} instead of the multiplicative one given to \cite{Quanrud24}.
Since this is a straightforward generalization of \cite{Quanrud24}, we focus on explaining how \cite{Quanrud24} works and why it would also work for additive approximations.

The framework of \cite{Quanrud24} reduces solving matroid intersection on a size-$n$ ground set to a sequence of the same problem on size-$O(r\log n/\eps)$ ground sets.
The high-level idea is to use a Multiplicate Weights Update approach similar to \cite{Assadi24} to assign initially uniform weight to each element.
In each iteration, we sample an $\Theta(r\log n/\eps)$-size subset based on the weight, over which we then solve matroid intersection.
The solutions we got from previous iterations tell us how much each element is ``covered'' by the accumulated solution, and elements that are ``covered'' less should be sampled with a larger probability in the future to ensure that overall all elements are taken into account sufficiently.
In particular, \cite{Quanrud24} considered the dual formulation of the matroid intersection LP:
\begin{equation}
\begin{split}
  &\text{minimize}\;\sum_{S \subseteq V}y_S \cdot \rank_1(S) + z_S \cdot \rank_2(S) \\
  &\text{subject to}\;\sum_{S \ni e}y_S + z_S \geq 1\;\text{for all}\;e \in V.
\end{split}
\label{eq:dual}
\end{equation}
Apart from solving the usual primal version of matroid intersection on the sampled subset, \cite{Quanrud24} additionally required a dual certificate $(A, B)$ that corresponds to $y_A = 1$ and $z_B = 1$ in \eqref{eq:dual} (with all other $y_S$ and $z_S$ being zero).
The pair $(A, B)$ is then extended to subsets of $V$ by taking the spans $A^\prime \defeq \spn_1(A)$ and $B^\prime \defeq \spn_2(B)$.\footnote{The extension step ensures that even though there are $2^n$ possible $U$, there are only roughly $n^r$ pairs $(A^\prime, B^\prime)$ that can be the ``effective'' output of the dual algorithm. This allows \cite{Quanrud24} to apply a union bound over all these possible outputs.}
Then, elements that are not covered by $(A^\prime, B^\prime)$, i.e., not in $A^\prime \cup B^\prime$, have their weights doubled\footnote{More specifically their weights get multiplied by $e$ in \cite{Quanrud24}.} so that they will more likely be sampled in future iterations.

In the end, after $O(\log n/\eps)$ iterations, \cite{Quanrud24} showed that one of the primal solutions $S$ has size at least $(1-O(\eps))r$.
To do so, there are two components.
First, in \cite[Lemma 9]{Quanrud24}, they showed that if we sample $\Theta(r\log n/\eps)$ elements with respect to the weights and obtain the dual $(A, B)$ from the subproblem, then the extended dual $(A^\prime, B^\prime)$ covers at least $(1-\eps)$-fraction of the weights with high probability.
This step remains the same for additive approximations because the dual solutions we compute are also feasible.
The second step is to show that at least one of the primal solutions has size at least $(1-O(\eps))r$.
This uses the following guarantee from the standard multiplicative-weight update algorithms.

\begin{lemma}[{\cite[Lemma 7]{Quanrud24}}]
  Consider a matrix $\boldsymbol{A}$ with non-negative coordinates and $m$ constraints.
  Suppose there is an oracle that, given weights $w \in \R_{\geq 0}^{m}$, computes a point $x$ and the list $I \defeq \{i \in [m]: (\boldsymbol{A}x)_i \geq 1\}$ in time $T$ such that $\sum_{i \in I}w_i \geq (1-\eps) \sum_{i \in [m]} w_i$.
  Then, there is a randomized $O\left(\frac{\log m}{\eps}(m + T)\right)$ time algorithm that generates a sequence of $L = O\left(\frac{\log m}{\eps}\right)$ weights $w_1, \ldots, w_L$ such that the average $\bar{x} \defeq (x_1 + \cdots + x_L) / L$ satisfies $\boldsymbol{A}\bar{x} \geq (1-O(\eps))1$ with high probability, where $x_i$ is the output of the oracle on input weights $w_i$.
  \label{lemma:mwu}
\end{lemma}

To apply \cref{lemma:mwu}, we let $\boldsymbol{A}$ be the constraint matrix of the dual formulation \eqref{eq:dual} of matroid intersection.
Note that the $m$ in \cref{lemma:mwu} corresponds to the number of constraints in $\boldsymbol{A}$ which equals to the number of elements $n$ for matroid intersection.
Let the oracle for \cref{lemma:mwu} does the following:
Given weights $w \in \R_{\geq 0}^{m}$, we sample $\Theta(r\log n/\eps)$ elements according to $w$, forming a subset $U \subseteq V$.
Then, we run the algorithm given to \cref{lemma:sparsification} to solve matroid intersection and its dual on $U$.
Let $S_U$ be the primal solution and $(A_U, B_U)$ be the dual.
By \cite[Lemma 9]{Quanrud24}, the extended dual $(A_U^\prime, B_U^\prime)$ covers $(1-\eps)$-fraction of the weights with high probability.
Conditioned on this happening for all iterations (which does so also with high probability), the dual variables induced by the $(A^\prime_U, B^\prime_U)$'s satisfy the requirement of \cref{lemma:mwu}.
After $L = O(\log n/\eps)$ iterations,  let $S_1, \ldots, S_L$ be the list of primal solutions and $(A_1, B_1), \ldots, (A_L, B_L)$ be the duals.
By the guarantee of the input oracle to \cref{lemma:sparsification}, we have
\[
  \frac{1}{L}\sum_{i \in [L]}|S_i| \geq \frac{1}{L}\sum_{i \in [L]}(\rank_1(A_i) + \rank_2(B_i)) - \Delta_{\mathcal{A}}.
\]
On the other hand, by \cref{lemma:mwu}, the average of $(A_i, B_i)$'s, after scaling up by a $(1+O(\eps))$ factor, is a feasible dual solution to the whole linear program.
By \cref{thm:matroid-intersection-dual}, this implies $\frac{1}{L}\sum_{i \in [L]}(\rank_1(A_i) + \rank_2(B_i)) \geq (1-O(\eps)) r$.
Therefore, the average of $|S_i|$ is at least $(1-O(\eps))r - \Delta_{\mathcal{A}}$, meaning that one of the common independent sets $S_i$ also has size at least $(1-O(\eps))r - \Delta_{\mathcal{A}}$.
This proves the correctness of \cref{lemma:sparsification}.
The running time follows by plugging in $T_{\mathcal{A}}(\Theta(r\log n/\eps))$ into \cref{lemma:mwu}.

\section{Parallel Weighted Matroid Intersection}

In this section we show how to generalize our parallel result of \cref{thm:parallel} to the weighted case.
Recall that in the weighted matroid intersection problem, we are additionally given an integer weight $w(e) \in \{0, \ldots, W\}$ for each element $e \in V$, and the goal is to find a common independent set $S \in \I_1 \cap \I_2$ maximizing $w(S) \defeq \sum_{e \in S}w(v)$.
By solving the unweighted version and obtaining the value of $r$ first (using our algorithm for \cref{thm:parallel}), we may assume that $\M_1$ and $\M_2$ share a common basis (by truncating the matroids to have rank $r$) and, moreover, each common independent set is contained in a common basis of the same weight (by adding dummy elements of weight $0$; see, e.g., \cite[Section 2]{Tu22} for more details).
This reduces the problem to finding a common basis $S$ maximizing $w(S)$.

We follow the weight-splitting framework of \cite{fujishige1995efficient,ShigenoI95,Tu22} which works as follows.
We \emph{split} the weight $w$ into $w_1: V \to \R$ and $w_2: V \to \R$, where a split $(w_1, w_2)$ is \emph{$\zeta$-approximate} if $w(e) \leq w_1(e) + w_2(e) \le w(e) + \zeta$ holds for all $e \in V$.
A common basis $S \in \I_1 \cap \I_2$ is \emph{$(w_1, w_2)$-optimal} if $S$ is a $w_1$-maximum basis in $\M_1$ and a $w_2$-maximum basis in $\M_2$ at the same time.
The outer loop of the algorithm is a scaling framework that has $O(\log (Wr))$ iterations, where in each iteration we start with a $2\zeta$-approximate split $(w_1^{(2\zeta)}, w_2^{(2\zeta)})$ and a $(w_1^{(2\zeta)}, w_2^{(2\zeta)})$-optimal $S^{(2\zeta)}$.
We then need to refine the current $(w_1^{(2\zeta)}, w_2^{(2\zeta)})$ and $S^{(2\zeta)}$ to a $\zeta$-approximate split $(w_1^{(\zeta)}, w_2^{(\zeta)})$ and a $(w_1^{(\zeta)}, w_2^{(\zeta})$-optimal $S^{(\zeta)}$.
Since $w_1 = W$ and $w_2 = 0$ is a trivial $W$-approximate split with any common basis being $(w_1, w_2)$-optimal, after $O(\log (Wr))$ iterations we get a $\frac{1}{2r}$-approximate split $(w_1^{*}, w_2^{*})$ and a $(w_1^{*}, w_2^{*})$-optimal $S^{*}$.
As the element weights are integral, this implies $S^{*}$ is a maximum weight common basis (see, e.g., \cite[Theorem 2.4]{Tu22}).
Thus in the remainder of the section we focus on how to refine a $2\zeta$-approximate split into an $\zeta$-approximate one.

\paragraph{Phase 1: Approximation Algorithm.}

Similar to many exact unweighted matroid intersection algorithms, the refinement process consists of two phases where the goal of the first one is to get a $w_1^{(\zeta)}$-maximum basis $S_1^{(\zeta)}$ in $\M_1$ and a $w_2^{(\zeta)}$-maximum basis $S_2^{(\zeta)}$ in $\M_2$ for some $\zeta$-approximate $(w_1^{(\zeta)}, w_2^{(\zeta)})$ such that $|S_1^{(\zeta)} \cap S_2^{(\zeta)}|$ is large (i.e., has size $r - \eps n$).
For this our auction algorithm of \cref{alg:auction} readily generalizes to the following \cref{alg:auction-weighted}.
There are a few key modifications.
First, instead of starting with $w_1(e) = 0$ and $w_2(e) = 0$ on Line~\ref{line:init}, we simply initialize these variables with the given $2\zeta$-approximate weight-split $(w_1^{(2\zeta)}, w_2^{(2\zeta)})$ by setting $w_1(e) \gets w_1^{(2\zeta)}(e)$ and $w_2(e) \gets w(e) - w_1(e)$.
Initially, $(w_1, w_2)$ is trivially a $\zeta$-approximate split.
To ensure that is stays $\zeta$-approximate, we also adjust the weight of each element by $\zeta$ instead of $1$ each time. 

\begin{algorithm}[!ht]
  \caption{Auction algorithm for weighted matroid intersection} \label{alg:auction-weighted}
  
  \SetEndCharOfAlgoLine{}

  \SetKwInput{KwData}{Input}
  \SetKwInput{KwResult}{Output}
  \SetKwInOut{State}{global}
  \SetKwProg{KwProc}{function}{}{}

  $p(e) \gets 0$ for all $e \in V$.\;
  $w_1(e) \gets w_1^{(2\zeta)}(e)$ and $w_2(e) \gets w(e) - w_1(e)$ for all $e \in V$.\; \label{line:init-weighted}
  $S_1 \gets$ a basis in $\M_1$ and $S_2 \gets$ a basis in $\M_2$.\;
  \While{true} { \label{line:iter-weighted}
    $X \gets \{e \in S_1 \setminus S_2: p(e) < 2/\eps\}$.\;
    \lIf{$|X| < \Delta$} {
      \textbf{return} $w_1$, $w_2$, $S_1$, and $S_2$.
    }
    \For{$e \in X$} {
      $p(e) \gets p(e) + 1$.\;
      \lIf{$w_1(e) + w_2(e) = w(e)$} {
          $w_2(e) \gets w_2(e) + \zeta$.
      }
      \lElse{
        $w_1(e) \gets w_1(e) - \zeta$.
      }
    }
    \For{$i \in \{1, 2\}$} {
      $S_i \gets$ a $w_i$-maximum basis in $\M_i$\;
      (breaking ties by preferring elements that were in the old $S_i$).\;
    }
  }
\end{algorithm}

We now prove a lemma analogous to \cref{lemma:meta}.

\begin{lemma}
  Given $\eps$ and $\Delta$, the output $(w_1, w_2)$ of \cref{alg:auction-weighted}  is a $\zeta$-approximate split with $S_i$ being $w_i$-maximum in $\M_i$ for each $i \in \{1, 2\}$.
  Crucially, we have $|S_1 \cap S_2| \geq r - (3 \eps r + \Delta)$.
  \label{lemma:meta-weighted}
\end{lemma}

Instead of the dual approach used for proving \cref{lemma:meta}, since we are in the weighted case now, we employ a more primal approach.
The following proof is virtually the same as in the original papers~\cite{fujishige1995efficient,ShigenoI95} (see \cite[Lemma 3]{ShigenoI95} or \cite[Lemma A.1]{Tu22}), but now also handling the $<\Delta$ remaining elements in $X$ at termination.

\begin{proof}
 That $(w_1, w_2)$ is $\zeta$-approximate and $S_i$ is $w_i$-maximum in $\M_i$ is straightforward from \cref{alg:auction-weighted}.
 Let $p(S) \defeq \sum_{e \in S}p(e)$.
 Letting $q(e) \defeq \lfloor p(e)/2 \rfloor$ be the number of times we adjust $w_1(e)$ (since for each element we start by adjusting $w_2$ and then alternate between the two operations), we have
 \begin{equation}
   q(e) = (w_1^{(2\zeta)}(e) - w_1(e))/\zeta \quad\text{and}\quad q(e) \geq (w_2(e) - (w(e) - w_1^{(2\zeta)}(e)))/\zeta - 1.
   \label{eq:q}
 \end{equation}
 Let $S^{(2\zeta)}$ be the common basis that is both $w_1^{(2\zeta)}$-maximum in $\M_1$ and $w_2^{(2\zeta)}$-maximum in $\M_2$ whose existence is guaranteed by the previous iteration of the scaling framework.
 It then follows that
 \begin{align*}
   q(S_1 \setminus S_2)
     &\stackrel{(i)}{=} q(S_1) - q(S_2) \stackrel{(ii)}{\leq} \frac{w_1^{(2\zeta)}(S_1) - w_1(S_1)}{\zeta} - \frac{w_2(S_2) - w(S_2) + w_1^{(2\zeta)}(S_2)}{\zeta} + |S_2| \\
     &\stackrel{(iii)}{\leq} \frac{w_1^{(2\zeta)}(S_1) - w_1(S^{(2\zeta)}) - w_2(S^{(2\zeta)}) + w(S_2) - w_1^{(2\zeta)}(S_2)}{\zeta} + r \\
     &\stackrel{(iv)}{\leq} \frac{w_1^{(2\zeta)}(S_1) - w(S^{(2\zeta)}) + w(S_2) - w_1^{(2\zeta)}(S_2)}{\zeta} + r \\
     &\stackrel{(v)}{\leq} \frac{w_1^{(2\zeta)}(S^{(2\zeta)}) - w(S^{(2\zeta)}) + w_2^{(2\zeta)}(S^{(2\zeta)})}{\zeta} + r \stackrel{(vi)}{\leq} 3r
 \end{align*}
 where (i) used \cref{claim:s2-minus-s1},
 (ii) used \eqref{eq:q},
 (iii) used that $S_i$ is $w_i$-maximum in $\M_i$ for each $i \in \{1, 2\}$ and $|S_2| = r$,
 (iv) used that $(w_1, w_2)$ is a $\zeta$-approximate split and thus $w_1(S^{(2\zeta)}) + w_2(S^{(2\zeta)}) \geq w(S^{(2\zeta)})$,
 (v) used that $w(S_2) - w_1^{(2\zeta)}(S_2) \leq w_2^{(2\zeta)}(S_2)$ since $(w_1^{(2\zeta)}, w_2^{(2\zeta)})$ is $2\zeta$-approximate and that $S^{(2\zeta)}$ is $w_i^{(2\zeta)}$-maximum in $\M_i$ for each $i \in \{1, 2\}$,
 and (vi) used again that $(w_1^{(2\zeta)}, w_2^{(2\zeta)})$ is $2\zeta$-approximate and thus $w_1^{(2\zeta)}(S^{(2\zeta)}) + w_2^{(2\zeta)}(S^{(2\zeta)}) \leq w(S^{(2\zeta)}) + 2\zeta r$.
 Since $q(v) = 1/\eps$ for all $v \in (S_1 \setminus S_2) \setminus X$, we have
 \[ |S_1 \cap S_2| = |S_1| - |S_1 \setminus S_2| \geq |S_1| - |X| - 3\eps r \geq r - (3 \eps r + \Delta), \]
 using that $|X| \leq \Delta$ when the algorithm terminates.
\end{proof}

Note that \cref{alg:auction-weighted} can be implemented in essentially the same way as \cref{alg:auction}.
As such, with \cref{cor:parallel-impl,lemma:meta-weighted} we conclude that in $O\left(\frac{n}{\eps \Delta}\right)$ rounds of rank queries or $O\left(\frac{n\sqrt{r}\log(n/r)}{\eps \Delta}\right)$ rounds of independence queries we can get two bases $S_1^{(\zeta)}$ and $S_2^{(\zeta)}$ with an intersection of $|S_1^{(\zeta)} \cap S_2^{(\zeta)}| \geq r - (3\eps r + \Delta)$ together with a $\zeta$-approximate split $(w_1^{(\zeta)}, w_2^{(\zeta)})$ that certify their (approximate) optimality.

\paragraph{Phase 2: Augmentation Algorithm.}

After Phase 1, we have two bases $S_1$ and $S_2$ with a large intersection such that each of them is $w_i$-maximum in $\M_i$.
The goal of the second phase is thus to ``augment'' $(S_1, S_2)$ so that they become equal.
This may involve possibly changing the weight split $(w_1, w_2)$ so that the new $S_i$ is $w_i$-maximum in $\M_i$, but we will maintain the fact that $w_1(e) + w_2(e)$ remains the same, and thus in the end when $S_1 = S_2$ we get a common basis which is $(w_1, w_2)$-optimal for some $\zeta$-approximate $(w_1, w_2)$.
At this point we will end the current scale with $(w_1^{(\zeta)}, w_2^{(\zeta)}) \defeq (w_1, w_2)$ and $S^{(\zeta)} \defeq S_1 = S_2$.
The augmentation step can be done by finding an augmenting path in a \emph{weighted} exchange graph (see, e.g., \cite{fujishige1995efficient,Tu22}).
Similar to the unweighted version, the weighted exchange graph (with respect to $(S_1, S_2)$) consists of $O(n^2)$ edges, each corresponding to an exchange relationship that can be computed by a single independence or rank query.
Therefore, in parallel, we can build this graph in a single round of queries, after which we can find a shortest augmenting path in the graph through which we can obtain the new $(S_1, S_2)$ with a larger intersection and the new weight split (see, e.g., \cite[Lemma 3.3]{Tu22}).

\paragraph{Putting Everything Together.}

Overall, in each scale, we spend either $O\left(\frac{n}{\eps \Delta}\right)$ rounds of rank queries or $O\left(\frac{n\sqrt{r}\log(n/r)}{\eps\Delta}\right)$ rounds of independence queries to refine the $2\zeta$-approximate solution to two bases with large enough intersection in Phase 1.
We can then augment them to a single common basis in $O(\eps r + \Delta)$ rounds of queries in Phase 2.
Setting $\eps$ and $\Delta$ the same as in the proof of \cref{thm:parallel}, we get the following parallel weighted matroid intersection algorithms.

\ParallelWeighted*

\end{document}